\documentclass[12pt, a4paper]{extarticle}
\usepackage[margin=2cm]{geometry}
\usepackage[affil-it]{authblk}

\usepackage[english]{babel}
\usepackage{graphicx} % Required for inserting images
\usepackage{wrapfig}
\usepackage{amsmath,amsthm,amssymb,scrextend}
\usepackage{booktabs}            % professional-quality tables
\usepackage{multirow}            % tabular cells spanning multiple rows
\usepackage{amsfonts}            % blackboard math symbols
\usepackage{multicol}
\usepackage[shortlabels]{enumitem}
\usepackage{cutwin}
\usepackage{float}
\usepackage{csquotes}
\usepackage{enumitem}
\usepackage{tikz}
\usepackage{hyperref}

\def \rR {\mathbb{R}}
\def \bf#1 {\textbf{#1 }}
\def \sumt {\sum\limits}
\providecommand{\keywords}[1]
{
  \small	
  \textbf{\textit{Keywords:}} #1
}

\renewenvironment{proof}{\begin{addmargin}[1em]{0em}\begin{newproof}}{\end{newproof}\end{addmargin}\qed}
\newtheorem{thm}{Theorem}
\newtheorem{lm}{Lemma}
\newtheorem{cor}{Corollary}
\newtheorem*{lm*}{Lemma}
\newtheorem{defin}{Definition}
\newtheorem*{defin*}{Definition}
\newtheorem{st}{Statement}

\def \sumt {\sum\limits}

\def \cN {\mathcal{N}}

\def \l {\lambda}

\def \rR {\mathbb{R}}
\newcommand{\mint}{\min\limits}

\begin{document}

\title{New centrality measure: ksi-centrality}
\author{Mikhail Tuzhilin%
  \thanks{Affiliation: Moscow State University, Electronic address: \texttt{mikhail.tuzhilin@math.msu.ru};}
}
\date{}
\maketitle

\begin{abstract}
We introduce new centrality measures, called ksi-centrality and normalized ksi-centrality measure the importance of a node up to the importance of its neighbors. First, we show that normalized ksi-centrality can be rewritten in terms of the Laplacian matrix such that its expression is similar to the local clustering coefficient. After that we introduce average normalized ksi-coefficient and show that for a random Erdos-Renyi graph it is almost the same as average clustering coefficient. It also shows behavior similar to the clustering coefficient for the Windmill and Wheel graphs. Finally, we show that the distributions of ksi centrality and normalized ksi centrality distinguish networks based on real data from artificial networks, including the Watts-Strogatz, Barabasi-Albert and Boccaletti-Hwang-Latora small-world networks. Furthermore, we show the relationship between normalized ksi centrality and the average normalized ksi coefficient and the algebraic connectivity of the graph and the Chegeer number.
\end{abstract}


%------------------------------------------------------------------

\section{Introduction}

One of the most important characteristics that distinguishes real-world networks (obtained from real data) from random networks is the average clustering coefficient. Networks that have a large average clustering coefficient and a small average shortest path are called small-world networks. In 1998, Watts and Strogatz found that most real-world networks have the small-world property or are small-world networks~\cite{Watts}, but random networks (the Erdos-Rényi graph) do not. In 1999, Albert, Jeong, and Barabasi gave another property that most real networks satisfy, but random networks are not called scale-free or power-law distributions~\cite{BA1}. Watts and Strogatz and Barabasi and Albert propose two models for constructing small-world networks~\cite{Watts},~\cite{BA2}. The problem was that the Watts-Strogatz network had a large average clustering coefficient, but this network did not satisfy the scale-free property. The opposite is true for the Barabasi-Albert network: this network is scale-free, but the average clustering coefficient is insufficient. Boccaletti, Hwang, and Latora~\cite{BH} propose a simple algorithm to construct a scale-free network with a large average clustering coefficient.

In this paper, we found a new centrality measure called ksi-centrality, whose normalized form has properties similar to the clustering coefficient, and whose distribution distinguishes real-world networks from artificial ones, including the Watts-Strogatz, Barabási-Albert and Boccaletti-Hwang-Latora models. As examples, we take real networks: social circles from Facebook~\cite{SN1} with 4039 nodes and 88234 edges, collaboration network of Arxiv General Relativity~\cite{SN2} with 5242 nodes and 14496 edges, LastFM Asia Social Network~\cite{SN3} with 7624 nodes and 27806 edges and C.elegans connectome~\cite{SN4} with 279 nodes and 2290 edges as well as artificial networks: Watts-Strogatz, Barabási-Albert, Boccaletti-Hwang-Latora and two Erdos-Renyi graphs with 4000 nodes.  

\section{Ksi-centrality and its properties}

Let's introduce the basic notations. Consider connected undirected graph $G$ with $n$ vertices. Denote by $A = A(G) = \{a_{ij}\}$ adjacency matrix of $G$ and by $L = L(G) = \{l_{ij}\}$ --- the Laplacian matrix. Let $\cN(i)$ denote the neighborhood of vertex $i$ (vertices adjacent to $i$), and let $d_i$ denote the degree of $i$. For any two disjoint subsets of vertices $H, K\subset V(G)$ denote the number of edges with one end in $H$ and another in $K$ by $E(H, K) = \big|{(v,w): v\in H,\ w\in K}\big|$.

Let's introduce new centrality called \emph{ksi-centrality}:
\begin{defin}
    For each vertex $i$ \bf{ksi-centrality} $\xi_{i}$ is the relation of the total number of neighbors of $i$'s neighbors excluding themselves divided by the total number of neighbors of $i$:
    $$  
        \xi_i = \xi(i) = \frac {\Big|E\big(\cN(i), V\setminus \cN(i)\big)\Big|} {\big|\cN(i)\big|} = \frac {\Big|E\big(\cN(i), V\setminus \cN(i)\big)\Big|} {d_i}.
    $$
\end{defin}

For fast calculations, the value of $\Big|E\big(\cN(i), V\setminus \cN(i)\big)\Big|$ can be found by multiplying the adjacency matrix by two columns of the adjacency matrix:

\begin{lm}
$$
    E\big(\cN(i), V\setminus \cN(i)\big) = \sumt_{j, k\in V(G)} a_{ij} a_{jk} \overline a_{ki},
$$ where $\overline a_{ki} = 1-a_{ki}.$
\end{lm}
\begin{proof}
    Let's fix $i$ and note that
    $$\sumt_{j\in V(G)} a_{ij} a_{jk} = \begin{cases}
        d_i, & k = i, \\
        1, & i\sim j\sim k, \\
        0, & \text{otherwise},
    \end{cases}
    \qquad\text{and}\qquad
    1 - a_{ki} = \begin{cases}
        1, & k = i, \\
        1, & k\not\sim i, \\
        0, & k\sim i.
    \end{cases}
    $$
    Therefore, 
    $$
        \Big|E\big(\cN(i), V\setminus \cN(i)\big)\Big| = d_i+ \big|k, j\in V(G): i\sim j\sim k, k\not\sim i\big| =  \sumt_{j, k\in V(G)} a_{ij} a_{jk} \overline a_{ki}.
    $$
\end{proof}

\begin{cor}\label{cor1} Let's $A$ be adjacency matrix of a graph. For each vertex $i$
    $$
        \xi_i = \frac {\Big(A^2\cdot\overline A\Big)_{ii}} {\big(A^2\Big)_{ii}},
    $$
    where $\overline{A} = I-A$ for $I$ --- matrix of all ones.
\end{cor}

Since $\frac {\Big|E\big(\cN(i), V\setminus \cN(i)\big)\Big|} {d_i} = \frac {d_i} {d_i} = 1$, when vertices of $\cN(i)\cup \{i\}$ have no adjacent vertices except themselves, we define ${\xi}_i = 1$ in the case when $d_i = 0$. Also note that our vertex $i\in V\setminus \cN(i)$, thus the ksi-centrality $\xi_i$ is always greater than or equal 1. Since the maximum number of edges from the neighborhood $\cN(i)$ to $V\setminus \cN(i)$ can be greater than $\Big|E\big(\cN(i), V\setminus \cN(i)\big)\Big|$. Let's give

\begin{defin}
    For each vertex $i$ \bf{normalized ksi-centrality} $\hat\xi_{i}$ is defined by following
    $$  
        \hat{\xi}_i = \hat{\xi}(i) = \frac {\Big|E\big(\cN(i), V\setminus \cN(i)\big)\Big|} {\big|\cN(i)\big|\cdot\big|V\setminus \cN(i)\big|} = \frac {\Big|E\big(\cN(i), V\setminus \cN(i)\big)\Big|} {d_i (n-d_i)}.
    $$
\end{defin}

It is easy to see that by this definition $\frac 1 {n-d_i}\leq\hat{\xi}_i\leq 1$. Since $\frac {\Big|E\big(\cN(i), V\setminus \cN(i)\big)\Big|} {d_i(n-d_i)} = \frac {d_i} {d_i(n-d_i)} = \frac 1 {n-d_i}$, when vertices of $\cN(i)\cup \{i\}$ have no adjacent vertices except themselves, we define $\hat{\xi}_i = \frac 1 n$ for the case, when $d_i = 0$. 

It is easy to see that by definition of ksi and normalized ksi centralities are connected to the Chegeer number.
\begin{st}\label{st1}  Consider an undirected graph $G$. Let $h(G)$ ne the Chegeer number of $G$.
\begin{enumerate}
    \item If for a vertex $i$ the degree $d_i\leq\frac n 2$, then $\xi_i\geq h(G)$,
    \item $\hat\xi_i\geq \begin{cases}
            h(G)\, (n-d_i), & \text{if }\, d_i\leq\frac n 2, \\
            h(G)\, d_i, & \text{otherwise,.}
    \end{cases}$
\end{enumerate}

\end{st}

%%%TO DO: check centrality axioms for ksi and ksi-hat

%%%TO DO: add transformation by deleting en edge

Let's remind the definition of local clustering coefficient $c_i$:
$$c_i = c(v_i) = \frac{2 \bigl|E\bigl(\cN(i)\bigr)\bigr|} {d_i(d_i-1)} = \frac {\sumt_{j, k\in V(G)} a_{ij} a_{jk} a_{ki}} {d_i (d_i-1)}$$.

This normalized version can be rewritten in a form similar to the local clustering coefficient, but in the terms of Laplacian matrix:

\begin{lm}\label{lm1}
    $$\hat\xi_i = \frac {\sumt_{j, k\in V(G)} l_{ij} l_{jk} l_{ki}} {d_i(n-d_i)} - \frac {d_i^2} {n-d_i}.$$
\end{lm}
\begin{proof}
    Let's rewrite the sum:
    $$
        \sumt_{j, k\in V(G)} l_{ij} l_{jk} l_{ki} = d_i  \sumt_{k\in V(G)} l_{ik} l_{ki} - \sumt_{j, k\in V(G), j\neq i} a_{ij} l_{jk} l_{ki} = d_i (d_i^2+d_i)-d_i \sumt_{j\in V(G),j\neq i} a_{ij} l_{ji} +
    $$
    $$
    +\sumt_{j, k\neq\in V(G), j\neq i, k\neq i} a_{ij} l_{jk} a_{ki} = d_i^3+d_i^2-d_i^2+\sumt_{j\in V(G),j\neq i} a_{ij} d_j a_{ji}-\sumt_{j, k\in V(G)} a_{ij} a_{jk} a_{ki} = d_i^3+\sumt_{j\in V(G):j\sim i} d_j-
    $$
    $$
    -\sumt_{j, k\in V(G)} a_{ij} a_{jk} a_{ki} = d_i^3+2 \big|E(\cN(i))\big|+\Big|E\big(\cN(i), V\setminus \cN(i)\big)\Big|-2\big|E(\cN(i))\big| = d_i^3+\Big|E\big(\cN(i), V\setminus \cN(i)\big)\Big|.
    $$
    By dividing to $d_i(n-d_i)$ the equality holds.
\end{proof}

Similarly, we define the average normalized ksi-coefficient for the entire graph $G$.
\begin{defin}
    The average normalized ksi-coefficient
    $$  
        \hat\Xi(G) = \frac 1 n \sumt_{i\in V(G)} \hat\xi_i.
    $$
\end{defin}

It turns out that for a random graph (Erdos-Renyi graph $(n,p)$) the expected value of normalized ksi-centrality equals almost $p$ and the expected value of the average normalized ksi-coefficient is almost $p$, as are the local clustering coefficient and the average clustering coefficient. To prove this, we first prove

\begin{thm}
    For any vertex $i\in V(G)$ in Erdos-Renyi graph $G(n,p)$ the expected number of
    $$
        \Big|E\big(\cN(i), V\setminus \cN(i)\big)\Big| = p(n-1)(1+p(1-p)(n-2)).
    $$
\end{thm}
\begin{proof}
    Let's denote the random variable $e = E(\cN(i), V\setminus\cN(i))$. First, let's note that $P(d_i = k) = C_{n-1}^kp^k(1-p)^{n-1-k}$. Since the maximum number of edges from $\cN(i)$ to $V\setminus \cN(i)\setminus\{i\}$ equal to $k(n-1-k)$, thus $P(e = t+k\,|\,d_i = k) = C_{k(n-1-k)}^t p^t (1-p)^{k(n-1-k)-t}$. Let's denote $f(k) = k(n-1-k)$. Thus, 
    $$
        E(e) = \sumt_{k = 0}^{n-1}\sumt_{t = 0}^{k(n-1-k)} (t+k)\, P(e = t+k) =  \sumt_{k = 0}^{n-1}\sumt_{t = 0}^{f(k)} (t+k)\, P(e = t+k\,|\,d_i = k)\, P(d_i = k) =
    $$
    $$
        = \sumt_{k = 0}^{n-1}\sumt_{t = 0}^{f(k)} (t+k)\, C_{f(k)}^t p^t (1-p)^{f(k)-t} C_{n-1}^kp^k(1-p)^{n-1-k} = 
    $$
    $$
       = \sumt_{k = 0}^{n-1} C_{n-1}^k p^k (1-p)^{n-k-1}\sumt_{t = 0}^{f(k)} (t+k)\,C_{f(k)}^t(1-p)^{f(k)-t} p^{t} = 
    $$
    $$
        = \sumt_{k = 0}^{n-1} C_{n-1}^k p^k (1-p)^{n-k-1}\Big(k+\sumt_{t = 1}^{f(k)} t\,C_{f(k)}^t(1-p)^{f(k)-t} p^{t}\Big)
    $$
    
    Note that $\sumt_{t = 0}^{f(k)}\,C_{f(k)}^t(1-p)^{f(k)-t} p^{t} = (p+1-p)^{f(k)} = 1$. Also $n(x+y)^{n-1}= \Big((x+y)^n\Big)_x = \Big(\sumt_{t = 0}^nC_{n}^t x^{t}y^{n-t}\Big)_x = \sumt_{t = 1}^n t C_{n}^t x^{t-1}y^{n-t}$, thus
    $$
        \sumt_{k = 0}^{n-1} C_{n-1}^k p^k (1-p)^{n-k-1}\Big(k+\sumt_{t = 1}^{f(k)} t\,C_{f(k)}^t(1-p)^{f(k)-t} p^{t}\Big) = \sumt_{k = 0}^{n-1} C_{n-1}^k p^k (1-p)^{n-1-k}\big(k+p f(k)\big) = 
    $$
    $$
        = p (n-1) + \sumt_{k = 0}^{n-1} C_{n-1}^k p^{k+1} (1-p)^{n-1-k} k (n-1-k) = 
    $$
    $$
        = p (n-1) + p^2(1-p) \sumt_{k = 1}^{n-2} C_{n-1}^k p^{k-1} (1-p)^{n-2-k} k (n-1-k) =p(n-1)+p^2(1-p)(n-1)(n-2),
    $$
    using the same procedure for $(n-1)(n-2)(x+y)^{n-3}= \Big((x+y)^{n-1}\Big)_{xy} = \sumt_{t = 1}^{n-2} t (n-1-t) C_{n-1}^t x^{t-1}y^{n-2-t}$.

\end{proof}

\begin{thm}\label{thm1}
    For any vertex $i\in V(G)$ in Erdos-Renyi graph $G(n,p)$ the expected number of
    $$
        \hat\xi_i =  p\Big(1-(1-p)^{n-1}\Big)+\frac {1-p^n} n, \qquad  \hat\Xi(G) =  p\Big(1-(1-p)^{n-1}\Big)+\frac {1-p^n} n.
    $$
\end{thm}
\begin{proof}
    Let's do the same calculations as in the previous theorem, but for $\frac {\Big|E\big(\cN(i), V\setminus \cN(i)\big)\Big|} {d_i (n-d_i)}$. Note that we defined $\hat\xi_i = \frac 1 n$, when $k = 0$. Thus, 

    $$
        E(\hat\xi_i) = \frac 1 n P(e = 0\,|\,d_i =0)\, P(d_i = 0)+\sumt_{k = 1}^{n-1}\sumt_{t = 0}^{f(k)} \frac {t+k} {k(n-k)}\, P(e = t+k\,|\,d_i = k)\, P(d_i = k) =
    $$
    $$
        = \frac {(1-p)^{n-1}} n+\sumt_{k = 1}^{n-1} C_{n-1}^k p^k (1-p)^{n-k-1}\sumt_{t = 0}^{f(k)} \frac {t+k} {k(n-k)}\,C_{f(k)}^t(1-p)^{f(k)-t} p^{t} =
    $$
    $$
        = \frac {(1-p)^{n-1}} n+\sumt_{k = 1}^{n-1} C_{n-1}^k p^k (1-p)^{n-k-1}\frac {k+p f(k)} {k(n-k)} = \frac {(1-p)^{n-1}} n+ \sumt_{k = 1}^{n-1} C_{n-1}^k p^k (1-p)^{n-1-k} \frac {1+p(n-1-k)} {n-k} =
    $$
    $$
          =\frac {(1-p)^{n-1}} n+ \sumt_{k = 1}^{n-1} C_{n-1}^k p^k (1-p)^{n-1-k} \frac {1+p(n-1-k)} {n-k} =\frac {(1-p)^{n-1}} n+ p-p(1-p)^{n-1}+
    $$
    $$
     + \sumt_{k = 1}^{n-1} C_{n-1}^k p^k (1-p)^{n-k}\frac 1 {n-k} = p-p(1-p)^{n-1}+ \sumt_{k = 0}^{n-1} \frac {(n-1)!} {(n-k)!\, k!} p^k (1-p)^{n-k} = 
    $$
    $$
    = p-p(1-p)^{n-1}+\frac 1 n\sumt_{k = 0}^{n-1} C_n^k p^k (1-p)^{n-k} =  p-p(1-p)^{n-1}+\frac {1-p^n} n.
    $$
    The same result for $\hat\Xi(G)$, since $\hat\Xi(G)$ is the average of  $\hat\xi_i$.
\end{proof}

We see that, if the number of vertices in the Erdos-Renyi graph $G(n,p)$ is large, then $\hat\Xi(G)\sim p$ is the same as the average clustering coefficient $C_{WS}(G)$. For a sparse Erdos-Renyi graph $G(n,p),\,p = \frac \l n$ the average normalized ksi-coefficient
$$
    \hat\Xi(G) = \frac \l n\Bigg(1-\Big(1-\frac \l n\Big)^{n-1}\Bigg)+\frac {1-\Big(\frac \l n\Big)^n} n = \frac {1+\l\big(1- e^{-\l}\big)} n + O\Big(\frac 1 {n^2}\Big),
$$

Therefore, it is asymptotically equivalent to the behavior of the average clustering coefficient $C_{WS}(G) = \frac \l n$. However, for real networks with a large number of vertices it may tend to 0 too (in some cases) due to division by $\frac 1 {n-d_i}$. Thus, average ksi-coefficient defined in the same way may be more useful for networks with a large number of vertices.

\begin{defin}
    The average ksi-coefficient
    $$  
        \Xi(G) = \frac 1 n \sumt_{i\in V(G)} \xi_i.
    $$
\end{defin}

\begin{thm}
    For any vertex $i\in V(G)$ in Erdos-Renyi graph $G(n,p)$ the expected number of
    $$
        \xi_i =  1+p(n-1)(1-p)\Big(1-(1-p)^{n-2}\Big), \qquad  \Xi(G) =  1+p(n-1)(1-p)\Big(1-(1-p)^{n-2}\Big).
    $$
\end{thm}
\begin{proof}
     Let's do the same calculations as in the previous theorem, but for $\frac {\Big|E\big(\cN(i), V\setminus \cN(i)\big)\Big|} {d_i}$. Note that we defined $\xi_i = 1$, when $k = 0$. Thus,
     $$
        E(\xi_i) = P(e = 0\,|\,d_i =0)\, P(d_i = 0)+\sumt_{k = 1}^{n-1}\sumt_{t = 0}^{f(k)} \frac {t+k} {k}\, P(e = t+k\,|\,d_i = k)\, P(d_i = k) =
    $$
    $$
        = {(1-p)^{n-1}}+\sumt_{k = 1}^{n-1} C_{n-1}^k p^k (1-p)^{n-k-1}\frac {k+p f(k)} {k} = {(1-p)^{n-1}}+ \sumt_{k = 1}^{n-1} C_{n-1}^k p^k (1-p)^{n-1-k} \big(1+p(n-1-k)\big) = 
    $$
    $$
        = 1+p(1-p)\sumt_{k = 1}^{n-2} C_{n-1}^k p^k (1-p)^{n-2-k} (n-1-k) = 1+p(1-p)\Big(n-1-(1-p)^{n-2}(n-1)\Big)
    $$
\end{proof}

We see that, for the Erdos-Renyi graph $G(n,p)$ with a large number of vertices, the average ksi-coefficient is $1+<k>(1-p)$, where $<k>$ is the average degree. For the sparse Erdos-Renyi graph ($p = \frac \l n$) 

$$
    \Xi(G) = 1+ \frac \l n (n-1)\Big(1-\frac \l n\Big)\Bigg(1-\Big(1-\frac \l n\Big)^{n-2}\Bigg) = 1+\l\big(1-e^{-\l})+O\Big(\frac 1 n\Big).
$$

Let's compare these coefficients for real networks and mathematical ones that model them.

\begin{enumerate}
    \item \bf{Ring lattice.} Consider a ring lattice or Watts-Strogatz network with $n$ vertices, $p = 0$ and $2k < n$ connections to each vertex. In this case, for each vertex $i$
    $$
        \Big|E\big(\cN(i), V\setminus \cN(i)\big)\Big| = 2k+2\sumt_{t = 1}^{k}t = 2k+2k(k+1) = k(k+3).
    $$
    Therefore, for each vertex $i$
    $$ 
        \hat\xi_i = \frac {k+3} {2(n-2k)},\qquad \xi_i = \frac {k+3} {2}, 
    $$
    and
    $$
        \hat\Xi(G) = \frac {k+3} {2(n-2k)},\qquad \Xi(G) = \frac {k+3} {2}.
    $$
    Thus, for the ring lattice $\hat\Xi(G)\rightarrow0$ for $n\rightarrow\infty$ and $\Xi(G)$ will be constant.
    
    \item \bf{Watts-Strogatz network.} Let's see how they are change for different parameters of a Watts-Strogatz network with $n$ vertices, some $p$ and $2k < n$. Let's denote the ksi-centrality and normalized ksi-centrality for the ring lattice ($p = 0$) by $\xi_0$ and $\hat\xi_0$ respectively (it is the same for all vertices). 
    
    In the figure~\ref{fig:1} we see that despite the fact that $\hat\xi_i\rightarrow0$ with $n\rightarrow\infty$ the spread of the relative value $\frac {\hat\xi_i} {\hat\xi_0}$ is almost the same as that of $\frac {\xi_i} {\xi_0}$, and their distributions are similar.
    In the figure~\ref{fig:2} we see a similar picture for the relative ksi-coefficient and normalized ksi-coefficient (they are almost the same, despite the fact that $\hat\Xi_i\rightarrow0$ for $n\rightarrow\infty$).
    
    Since the normalized ksi-coefficient tends to 0 with increasing $n$, it is better to use the ksi-coefficient for large networks. In the figure~\ref{fig:3} we see that the distribution of the relative ksi-coefficient up to rewiring probability $p$ is almost the same for different number of vertices $n = 200, 500, 1000, 2000$.

\begin{figure}[h!]
    \centering
	\includegraphics[height = 0.9\textheight, width = 1.0\textwidth]{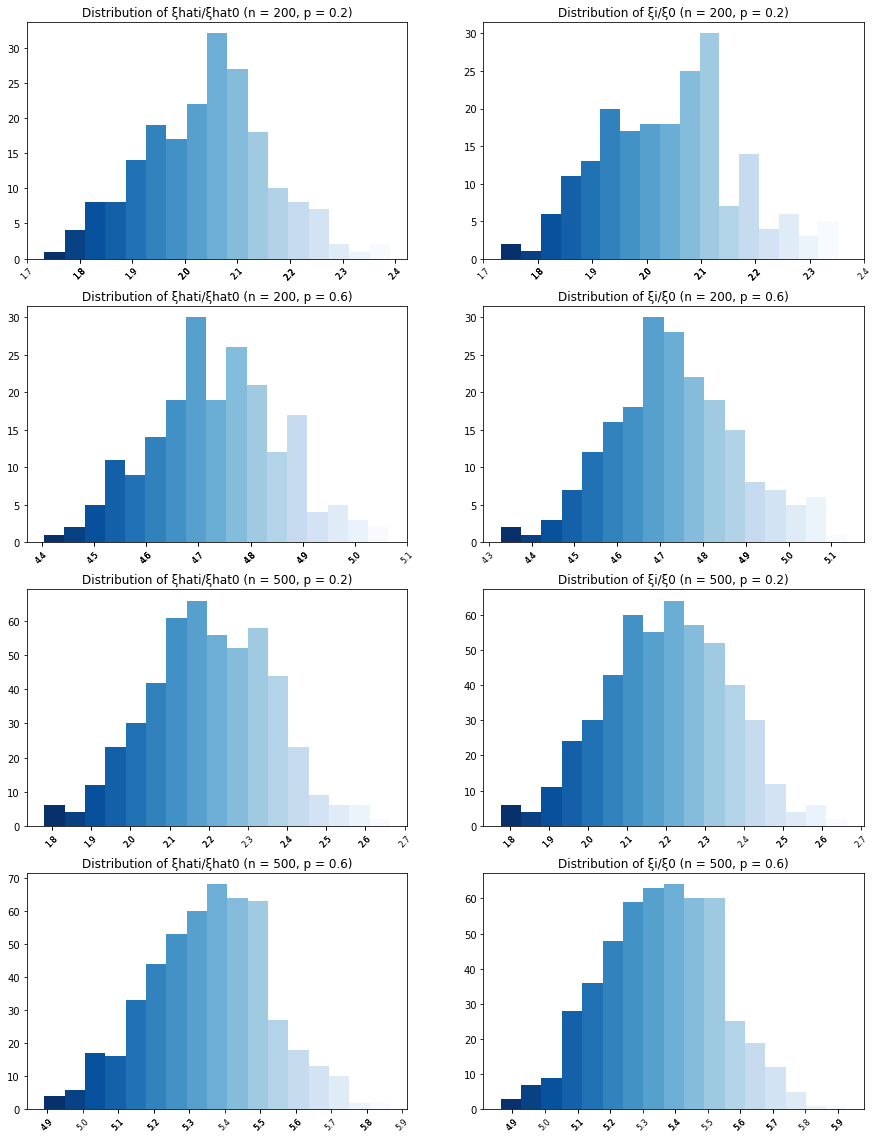}
	\caption{Comparison of distributions of $\frac {\hat\xi_i} {\hat\xi_0}$ and $\frac {\xi_i} {\xi_0}$ for the numbers of Watts-Strogatz network vertices $n = 200, 500$, rewiring probabilities $p = 0.2, 0.6$ and $2k = 100$ --- the value corresponded to initial degree.}
	\label{fig:1}
\end{figure}

\begin{figure}[h!]
    \centering
	\includegraphics[width = 0.88\textwidth]{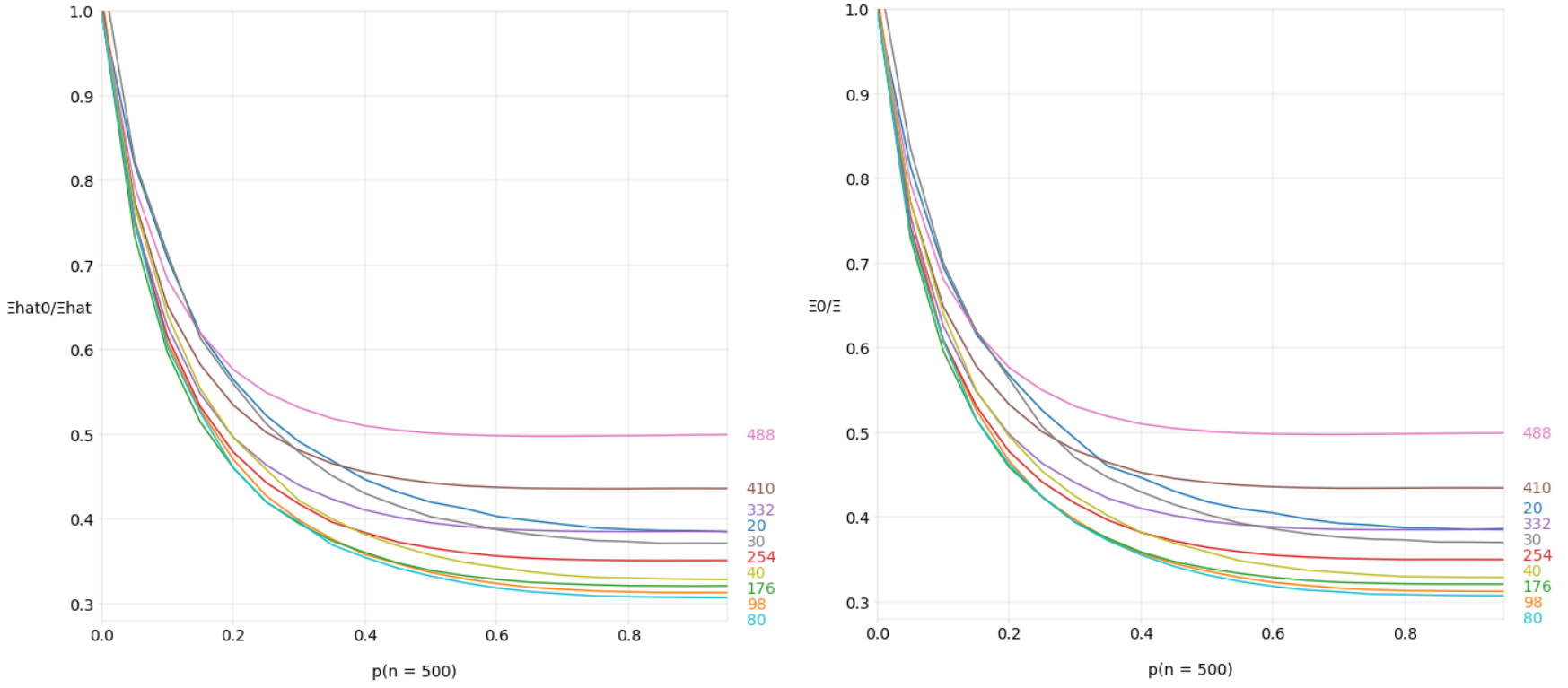}
	\caption{Comparison of ratios $\frac {\hat\Xi(G_0)} {\hat\Xi(G_p)}$ and $\frac {\Xi(G_0)} {\Xi(G_p)}$ for the Watts-Strogatz network $G_p$ with $n = 500$, rewiring probability $p$ and the value corresponded to initial degree $2k$ on the right side of each plot.}
	\label{fig:2}
\end{figure}

\begin{figure}[h!]
    \centering
    \vspace{-3pt}
	\includegraphics[width = 0.9\textwidth]{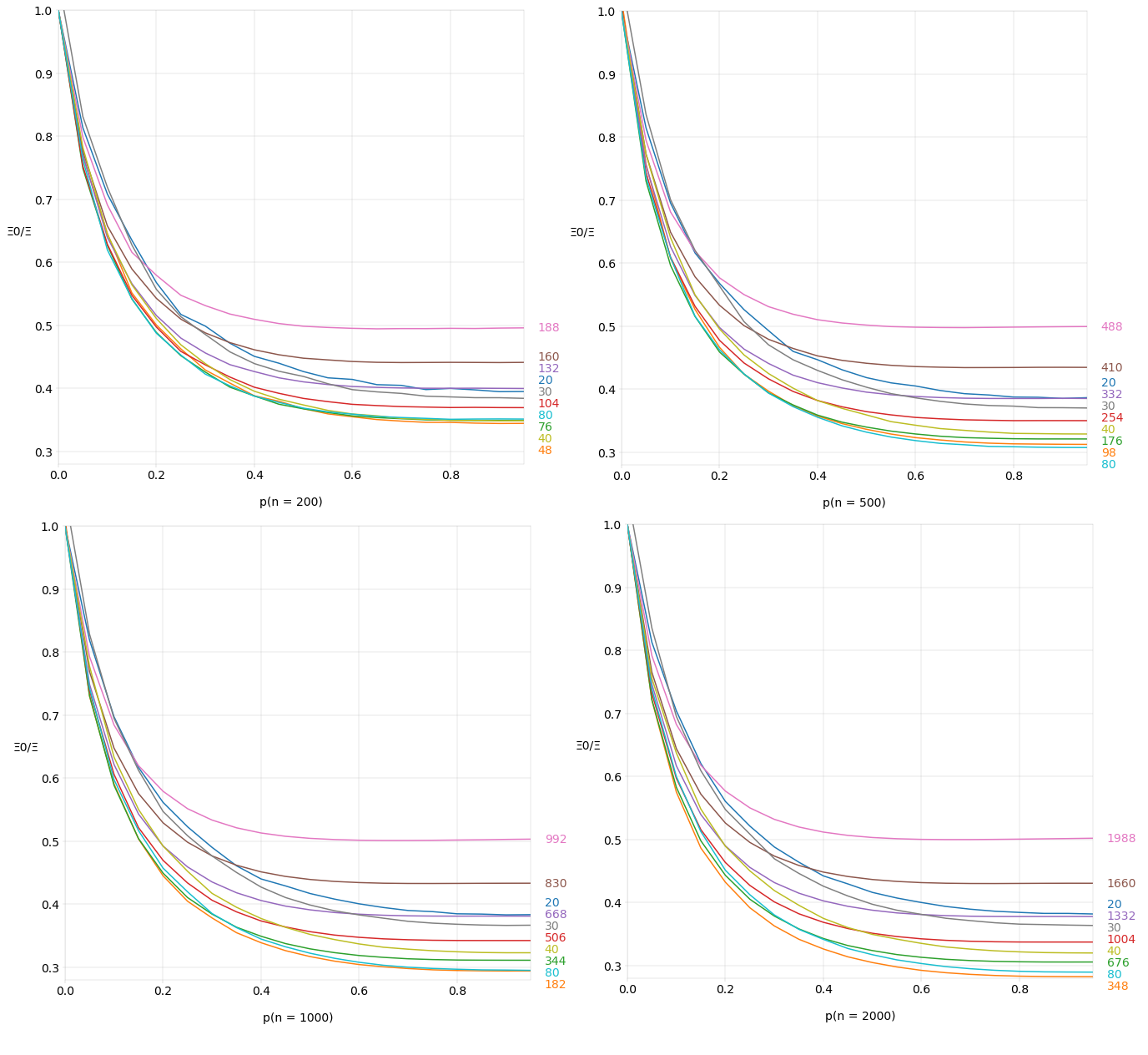}
	\caption{Comparison of ratios $\frac {\Xi(G_0)} {\Xi(G_p)}$ for the Watts-Strogatz network $G_p$ with $n =200, 500, 1000, 2000$, rewiring probability $p$ and the value corresponded to initial degree  $2k$ on the right side of each plot.}
	\label{fig:3}
\end{figure}

    \item \bf{Barabasi-Albert network.} We compare them for Barabasi-Albert network with $n$ vertices and $k$ edges that are preferentially attached. In the figure~\ref{fig:4} we see that for Barabasi-Albert network the distributions of $ {\hat\xi_i}$ and $\xi_i$ are not so similar. However, they are similar for different number of network vertices $n = 200, 500$ respectively up to the same ratio of preferentially attached edges to $n$. 
    
    In the figure~\ref{fig:5} we calculated the normalized ksi-coefficient and the ksi-coefficient for 8 groups with different parameters $k = \frac {n} {30}, \frac {5n} {30}, \frac {9n} {30}, ...,  \frac {29n} {30}$ (which depend on $n$). In each group we changed the number of vertices $n = 200, 500, 750, 1000, 1500, 2000$. It turns out that the normalized ksi-coefficient hardly changed for different numbers of vertices, the and ksi-coefficient increased with $n$.

\begin{figure}[h!]
    \centering
	\includegraphics[height = 0.9\textheight, width = 1.0\textwidth]{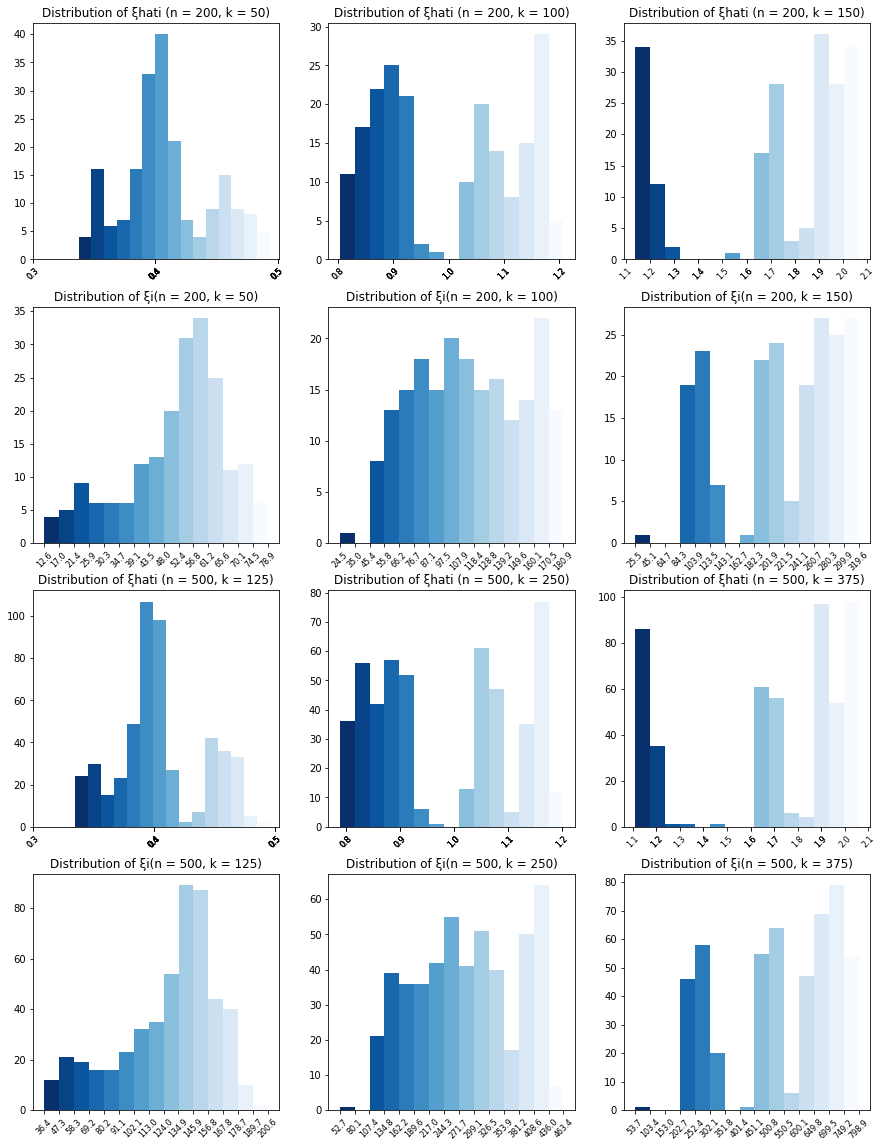}
	\caption{Comparison of distributions $ {\hat\xi_i}$ and $\xi_i$ for vertices of Barabasi-Albert network with $n = 200, 500$ vertices and preferentially attached edges $k = \frac n 4, \frac {n} 2, \frac {3n} 4$.}
	\label{fig:4}
\end{figure}

\begin{figure}[h!]
    \centering
	\includegraphics[width = 1.0\textwidth]{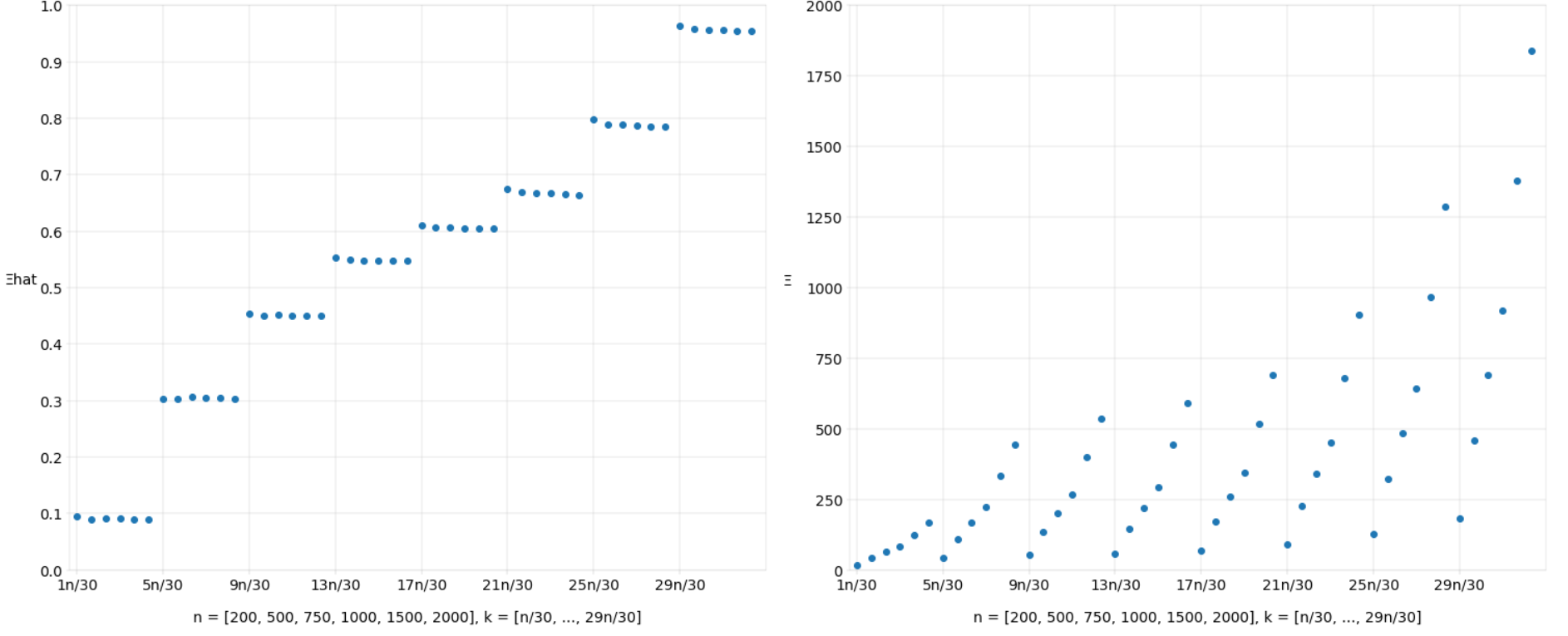}
	\caption{Comparison of distributions $ {\hat\Xi}$ and $\Xi$ for vertices of Barabasi-Albert network with $n = 200, 500, 750, 1000, 1500, 2000$ vertices. Each number of vertices corresponds to 6 consequent points in each group, and the number of relative preferentially attached edges $k = \frac {n} {30}, \frac {5n} {30}, \frac {9n} {30}, ...,  \frac {29n} {30}$ corresponds to a group respectively.}
	\label{fig:5}
\end{figure}

    \item \bf{Another real-data networks.} We compare the distributions of normalized ksi-centrality and ksi-centrality for different networks: social circles from Facebook~\cite{SN1} (https://snap.stanford.edu/data/ego-Facebook.html), collaboration network of Arxiv General Relativity~\cite{SN2} (https://snap.stanford.edu/data/ca-GrQc.html), LastFM Asia Social Network~\cite{SN3} (https://snap.stanford.edu/data/feather-lastfm-social.html), C.elegans connectome~\cite{SN4} (https://www.wormatlas.org/neuronalwiring.html) and Barabasi-Albert (4000, 43), Watts-Strogatz (4000, 21, 0.3),  Erdos-Renyi (4000, 0.2), Erdos-Renyi (4000, 0.001) networks with similar parameters. First, we see again that the distributions of normalized ksi-centrality and ksi-centrality are similar (figures~\ref{fig:6}, \ref{fig:8} and \ref{fig:7}, \ref{fig:9}), thus it is better to use ksi-centrality for the calculation since normalized ksi-centrality can be very small. Also, in figures~\ref{fig:6},\ref{fig:7} and \ref{fig:8},\ref{fig:9} we see that distributions of both ksi-centralities distinguish real networks (from real data) from artificial networks and they do not depend on the degree distribution (all networks except Watts-Strogatz network and Erdos-Renyi network have power-law degree distribution and the last have normal degree distribution). For real networks, the distributions of both ksi-centralities are right-skewed and for artificial are centered. We also we calculated normalized ksi-coefficient and ksi-coefficient for these networks (see the table~\ref{tab:1}).

    \begin{table}[h!]
        \centering
        \begin{tabular}{|c|c|c|c|}
        \hline
             Network & $\hat\Xi$ & $\Xi$ & $n$ \\
        \hline
             Facebook & 0.0202 &	81.1540 & 4039 \\
             Collaboration & 0.0013 & 6.9742 & 5242 \\
             LastFM & 0.0029 & 21.8505 & 7624 \\
             C.elegans & 0.0885 & 23.3842 & 279 \\
             Barabasi-Albert & 0.0355 & 138.9953 & 4000 \\
             Watts-Strogatz & 0.0039 & 15.6413 & 4000 \\
        \hline
        \end{tabular}
        \caption{Normalized ksi-coefficient and ksi-coefficient for different networks: social circles from Facebook, collaboration network of Arxiv General Relativity, c.elegans connectome, Barabasi-Albert (4000, 43) and Watts-Strogatz networks (4000, 21, 0.3).}
        \label{tab:1}
    \end{table}

    \item \bf{Boccaletti-Hwang-Latora network.} Finally, we compare them for Boccaletti-Hwang-Latora network with initial vertices $n_0 = 500$, the smallest initial degree $m = 50$ and the number of vertices $n= 4000$. We first compute a random degree distribution with degrees $d_i = m, m+1, ..., (n_0-m)$ for $n_0$ vertices and construct the Havel-Hakimi graph~\cite{HH} with corresponded degree sequence. Then, we use Boccaletti-Hwang-Latora algorithm to construct final graph. 

    In the figure~\ref{fig:10} we see that according to the distribution of ksi and normalized ksi-centrality, the Boccaletti-Hwang-Lator network belongs to the class of artificial networks as well.

\begin{figure}[H]
    \centering
	\includegraphics[width = 0.85\textwidth]{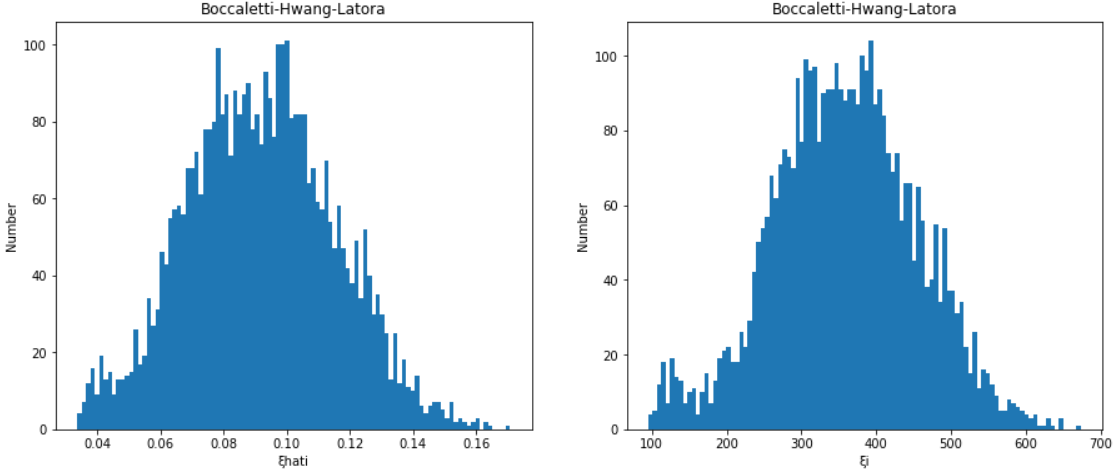}
	\caption{Distributions of ${\hat\xi_i}$ and ${\hat\xi_i}$ Boccaletti-Hwang-Latora network with initial vertices $n_0 = 500$, the smallest initial degree $m = 50$ and the number of vertices $n= 4000$.}
	\label{fig:10}
\end{figure} 

\begin{figure}[h!]
    \centering
	\includegraphics[width = 1.0\textwidth]{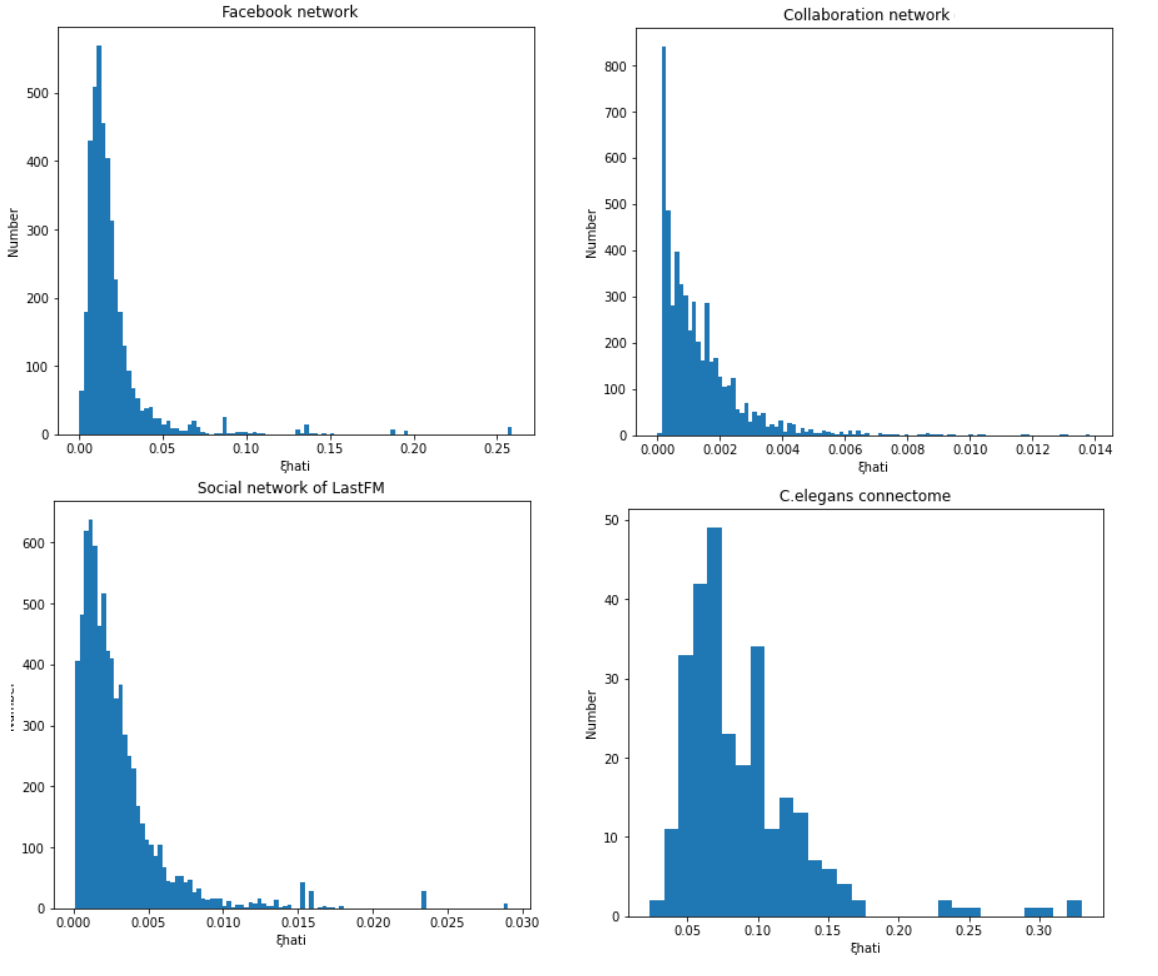}
	\caption{Distribution of ${\hat\xi_i}$ for different real networks: social circles from Facebook, collaboration network of Arxiv General Relativity, c.elegans connectome.}
	\label{fig:6}
\end{figure} 

\begin{figure}[h!]
    \centering
	\includegraphics[width = 1.0\textwidth]{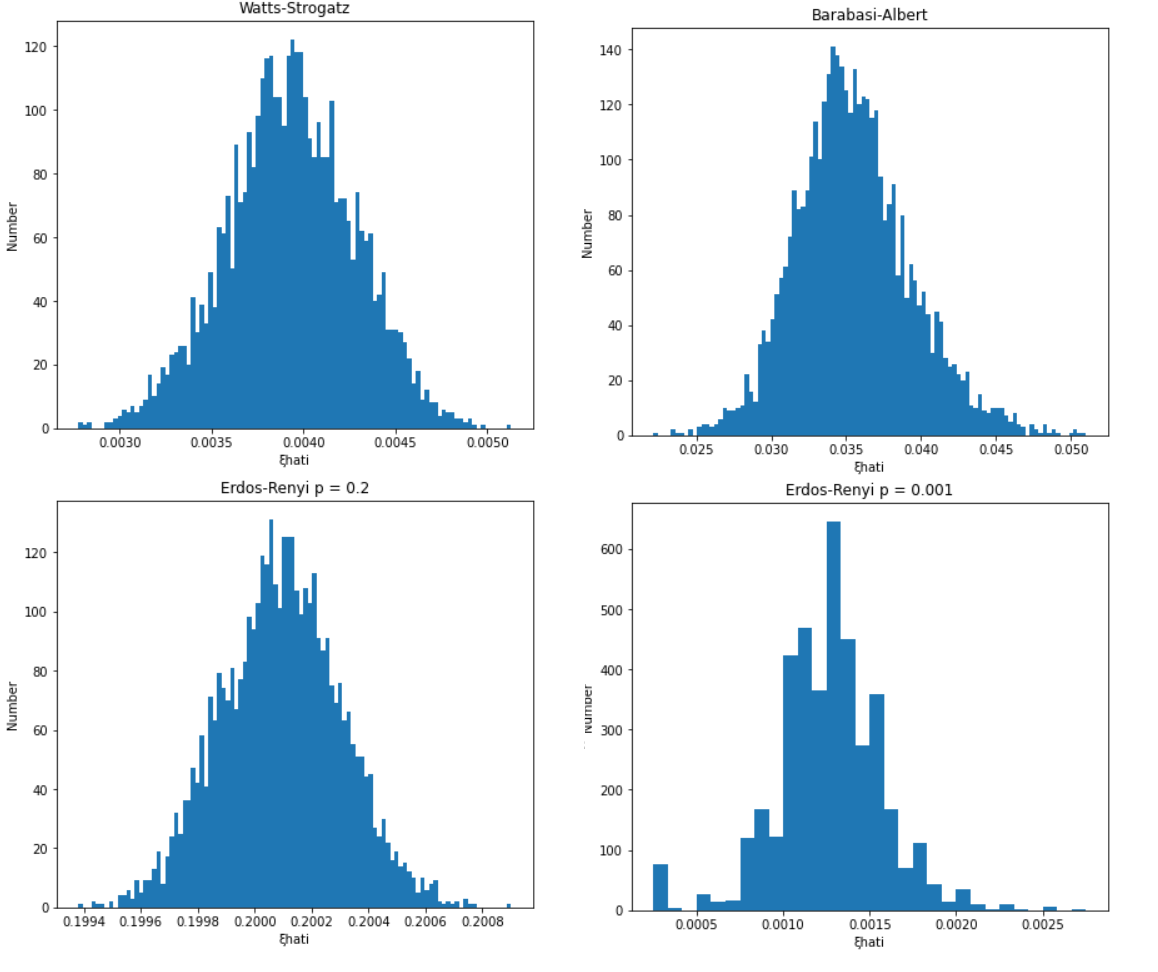}
	\caption{Distribution of ${\hat\xi_i}$ for different artificial networks: Barabasi-Albert (4000, 43), Watts-Strogatz (4000, 21, 0.3), Erdos-Renyi (4000, 0.2) and Erdos-Renyi (4000, 0.001) networks.}
	\label{fig:7}
\end{figure} 

\begin{figure}[h!]
    \centering
	\includegraphics[width = 1.0\textwidth]{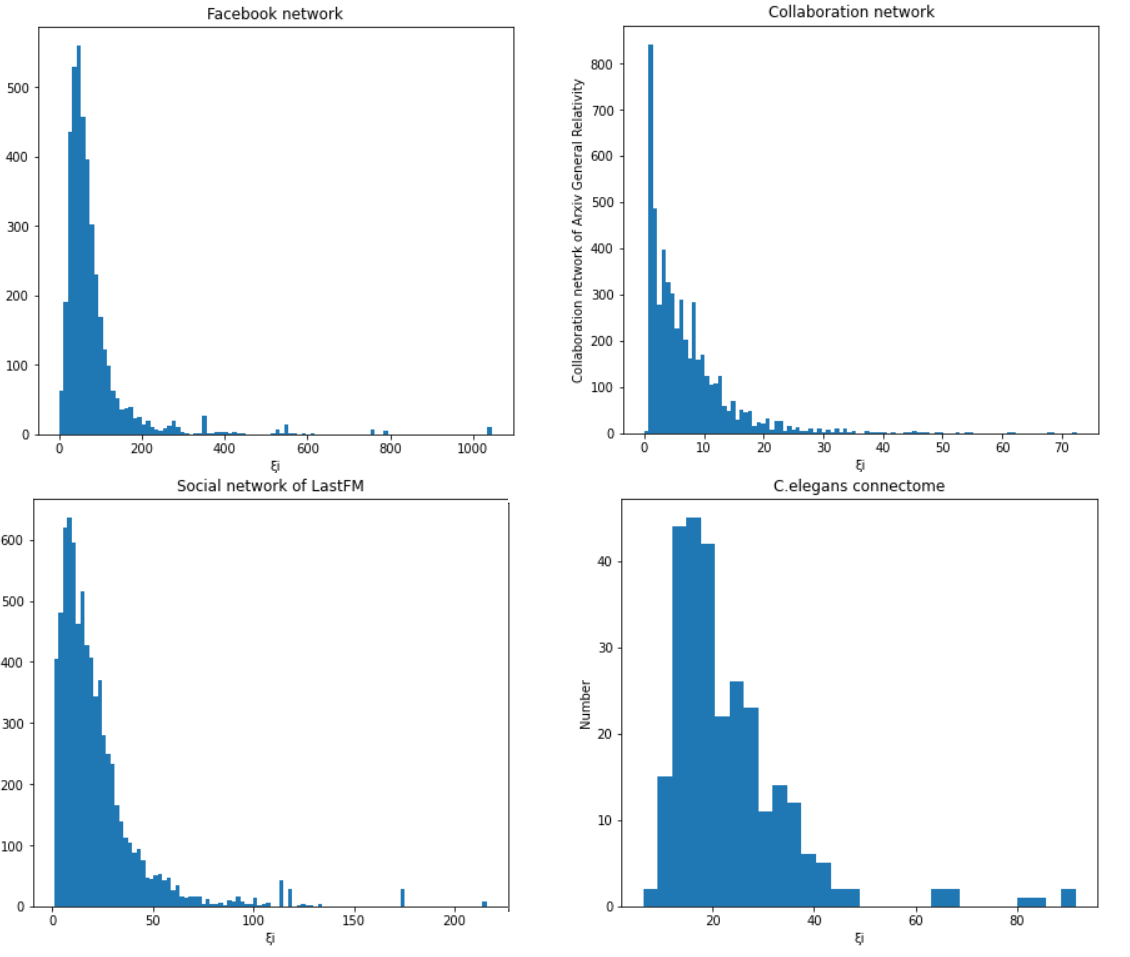}
	\caption{Distribution of $\xi_i$ for different real networks: social circles from Facebook, collaboration network of Arxiv General Relativity, c.elegans connectome.}
	\label{fig:8}
\end{figure} 

\begin{figure}[h!]
    \centering
	\includegraphics[width = 1.0\textwidth]{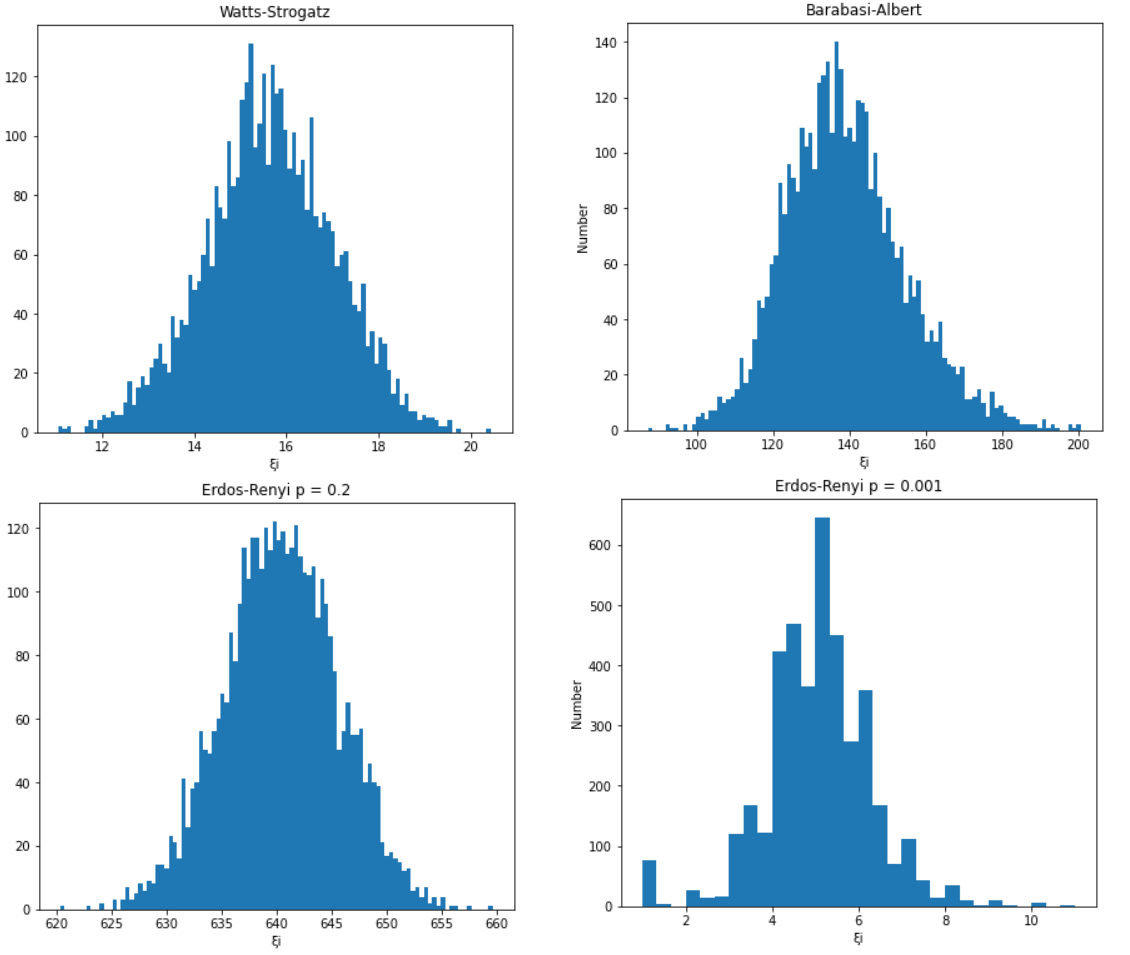}
	\caption{Distributions of $\xi_i$ for different artificial networks: Barabasi-Albert (4000, 43), Watts-Strogatz (4000, 21, 0.3), Erdos-Renyi (4000, 0.2) and Erdos-Renyi (4000, 0.001) networks.}
	\label{fig:9}
\end{figure} 

\end{enumerate}

Other surprising result is that the normalized ksi-centrality (as well as the average normalized ksi-coefficient) is related to the $\l_2$ algebraic connectivity (or the second eigenvalue of the Laplacian matrix).

\begin{thm}\label{thm2}
    Let's $G$ --- undirected graph with $n$ vertices and Laplacian matrix $L$. For any vertex $i\in V(G)$
    $$
    \hat\xi_i\geq \frac {\l_2} n, \qquad
    \hat\Xi(G)\geq \frac {\l_2} n.
    $$
\end{thm}
\begin{proof}
    Let's remind that $\l_2 = \mint_{x\in \rR^n, (x,\bf{1}) = 0} \frac {(Lx, x)} {(x,x)} = \mint_{x\in \rR^n, (x,\bf{1}) = 0} \frac {\sumt_{i,j\in V(G), i\sim j}(x_i-x_j)^2} {(x,x)},$ where \bf{1} is the vector of ones and $(\cdot,\cdot)$ is the standard dot product in $\rR^n$. Let's define for the vertex $i$ a vector
    $$
        y = \big(y_j\big) = \begin{cases}
            n-d_i, & j\in \cN(i), \\
            -d_i, & j\in V(G)\setminus\cN(i).
        \end{cases}
    $$
    It's easy to see that $(y, \bf{1}) = 0$ and $(y,y) = (n-d_i)^2d_i+d_i^2(n-d_i) = d_i(n-d_i)n$. Also $\sumt_{k,j\in V(G), k\sim j}(y_k-y_j)^2 = n^2 \Big|E\big(\cN(i), V\setminus\cN(i)\big)\Big|.$ Therefore,
    $$
        \l_2\leq \frac {\sumt_{k,j\in V(G), k\sim j}(y_i-y_j)^2} {(y,y)} = n \frac {\Big|E\big(\cN(i), V\setminus\cN(i)\big)\Big|} {d_i(n-d_i)} = n\,\hat \xi_i.
    $$
\end{proof}

\section{Another examples}

\begin{enumerate}
    \item \bf{Star graph.} Let's consider the star graph with $n+1$ vertices. For this graph
    $$
        \hat\xi_i = 1, \qquad \xi_i = \begin{cases}
            1, & \text{if $i$ central vertex,} \\
            n, & \text{otherwise,}
        \end{cases}
    $$
    and thus,
    $$
        \hat\Xi(G) = 1,\qquad \Xi(G) = \frac {n^2+1} {n+1}\sim n.
    $$
    We see that normalized ksi-centrality does not distinguish between the vertices of the star graph (as the local clustering coefficient) and its values are equal to the value of an isolated vertex. However, for the ksi-centrality, the vertices on periphery are more important, since they have a significant central neighbor. Also, the average normalized ksi-coefficient is constant, but the average ksi-coefficient tends to infinity with an increase in the number of vertices.
    
    \item \bf{Windmill graph.} Let's consider the windmill graph $W(n,k)$ consisting of $n$ copies of the complete graph $K_k$ connected to the central vertex. For this graph
    $$
        \hat\xi_i = \begin{cases}
            1, & \text{if $i$ central vertex,} \\
            \frac n {nk+1-k}, & \text{otherwise,}
        \end{cases} \qquad \xi_i = \begin{cases}
            1, & \text{if $i$ central vertex,} \\
            n, & \text{otherwise,}
        \end{cases}
    $$
    and thus
    $$
        \hspace{-10pt}\hat\Xi\big(W(n,k)\big) = \frac 1 {nk+1} \Big(1+\frac {n^2} {nk+1-k}\Big) = \frac {n^2+nk+1-k} {(nk+1)(nk+1-k)}\sim \frac 1 {k^2}, \qquad \Xi\big(W(n,k)\big) = \frac {1+n^2k} {nk+1}\sim n.
    $$
    We see that for ksi-centrality the windmill graph $W(n,k)$ is the same as star graph, the opposite holds for normalized ksi-centrality where the central vertex is more important than others for a large number of vertices in windmill graph. The ksi-coefficient is the same as for the star graph and is proportional to $n$. The normalized ksi-coefficient tends to $\frac 1 {k^2}$ for $n\rightarrow\infty$. Note that the average clustering coefficient $C_{WS}\big(W(n,k)\big)\rightarrow 1$ for $n\rightarrow\infty$.
    
    \item \bf{Wheel graph.} Let's consider the wheel graph $W(n)$ with $n+1$ vertices. For this graph
    $$
        \hat\xi_i = \begin{cases}
            1, & \text{if $i$ central vertex,} \\
            \frac {n+2} {3(n-2)} = \frac 1 3+\frac {4} {3(n-2)}, & \text{otherwise,}
        \end{cases} \qquad \xi_i = \begin{cases}
            1, & \text{if $i$ central vertex,} \\
            \frac {3+2+n-3} {3} = \frac {n+2} 3, & \text{otherwise,}
        \end{cases}
    $$
    and thus
    $$
        \hat\Xi\big(W(n)\big) = \frac 1 {n+1}\Bigg(1+\frac {n(n+2)} {3(n-2)}\Bigg) = \frac {(n+6)(n-1)} {3(n+1)(n-2)}\rightarrow \frac 1 3,
    $$
    $$
        \Xi\big(W(n)\big) = \frac {1+\frac {n^2+2n} 3} {n+1} = \frac {n^2+2n+3}{3(n+1)}\sim\frac {n+1} 3.
    $$
    We see that normalized ksi-coefficient tends to $\frac 1 3$ with $n\rightarrow\infty$. Let's note that average clustering coefficient $C_{WS}\big(W(n)\big)\rightarrow\frac 2 3$ for $n\rightarrow\infty$.
    \item \bf{Nested triangles graph.} Let's consider the nested triangles graph $T(n)$ with $n$ triangles and $3n$ vertices. Let's enumerate the nested triangles with $T_1, T_2, ... T_{n}$ by inclusion. For this graph
    $$
        \hat\xi_i = \begin{cases}
            \frac 8 {9(n-1)}, & \text{if $i\in T_1$ or $i\in T_{n}$ ,} \\
            \frac {13} {4(3n-4)}, & \text{if $i\in T_2$ or $i\in T_{n-1}$,} \\
            \frac {7} {2(3n-4)}, & \text{otherwise,}
        \end{cases}
        \qquad \xi_i = \begin{cases}
            \frac {4+2+2} {3} = \frac 8 3, & \text{if $i\in T_1$ or $i\in T_{n}$ ,} \\
            \frac {3+4+3+3} {4} = \frac {13} 4, & \text{if $i\in T_2$ or $i\in T_{n-1}$,} \\
            \frac {4+4+3+3} {4} = \frac {7} 2, & \text{otherwise.}
        \end{cases} 
    $$
    and thus,
    $$
        \hat\Xi\big(T(n)\big) = \frac 3 {3n} \bigg(\frac {16} {9(n-1)}+\frac {13} {2(3n-4)}+\frac {7(n-4)} {2 (3n-4)}\bigg) = \frac {63 n^2 - 102 n + 7 } {18n(n-1)(3n-4)}\sim\frac 7 {6n},
    $$
    $$
        \Xi\big(T(n)\big) = \frac 3 {3n} \bigg(\frac {16} {9}+\frac {13} {2}+(n-4)\frac {7} {2 }\bigg) = \frac {63n-103} {18 n}\rightarrow\frac 7 2.
    $$
    We see that since in the nested triangles graph the structure of $E\big(\cN(i), V\setminus\cN(i)\big)$ is practically the same for every vertex $i$ and does not depend on $n$, then $\hat\Xi(G)\rightarrow0$ for $n\rightarrow\infty$ and in this case $\Xi(G)$ is more informative. 
\end{enumerate}

\section{Discussion}

In this article we proposed a new measure of centrality called ksi-centrality. This centrality by definition can identify an important node based on the power of its neighbors, even if its neighbors do not know each other. A node with high ksi-centrality can be a ruler who has many contacts, and his contacts also have many contacts or a gray suit who has contacts with the most powerful people, and they may not know each other.

In a star graph, the periphery vertices are more important than the central vertex for ksi-centrality, because their neighbor is more ``powerful'' than the neighbors of the central vertex. Therefore, it does not satisfy the Freeman star property~\cite{Free}. The ksi-centrality distributions for a star graph and a windmill graph are the same, because these graphs are ``similar'' in terms of neighbor structure. For a wheel graph, the situation is the same: the periphery vertices are more important. For a nested triangle graph, the most important vertices are the internal vertices, because they have better neighbors

We have shown that ksi-centrality can be easily computed (corollary~\ref{cor1}), its normalized form has many interesting properties: it can be rewritten in a form similar to clustering coefficient (Lemma~\ref{lm1}), the average normalized ksi-coefficient has almost the same value as the average clustering coefficient for the Erdős-Rényi graph (theorem~\ref{thm1}), it is related to the algebraic connectivity (theorem~\ref{thm2}) and the Chegeer number of a graph (statement~\ref{st1}), and for the Barabási-Albert network it depends only on the ratio of preferentially attached edges to the number of vertices, but is independent of the network size. It also exhibits behavior similar to the clustering coefficient for mathematical graphs: the windmill graph and the wheel graph. Also, the normalized ksi-coefficient is the same for each vertex of the star graph, like the local clustering coefficient. However, for real networks (with a large number of nodes), the normalized ksi-coefficient can be very small. We have shown that ksi-centrality has a very similar distribution and thus will be more useful for computations in applications.

We show that for real networks: Facebook, Collaboration network of Arxiv General Relativity, LastFM Asia Social Network, and C.elegans connectome, the distribution of ksi centrality and normalized ksi centrality resembles a right-skewed normal distribution, while for the artificial networks of Barabasi-Albert, Watts-Strogatz, Boccaletti-Hwang-Latora and Erdos-Renyi, it resembles a central normal distribution. Thus, this distribution can distinguish real networks from artificial ones regardless of the degree distribution (Barabasi-Albert and Boccaletti-Hwang-Latora are a scale-free networks, while Watts-Strogatz is not). As for the average ksi coefficient and normalized ksi coefficient, they did not show significant results on real data networks. Perhaps, their definitions can be modified to make them more effective in applications. As a result, ksi centrality is not only a useful tool for analyzing social networks and real data networks, but also interesting from a theoretical (mathematical) point of view, and further research is needed to understand the role of its average coefficients in applications.

\end{document}